\documentclass[11pt]{article}

\usepackage[letterpaper, left=1in, right=1in, top=1in,bottom=1in]{geometry}

\usepackage{geometry}

\usepackage{parskip}
\usepackage{nicefrac}
\usepackage{booktabs} 

\usepackage[utf8]{inputenc}

\usepackage{graphpap,amscd,mathrsfs,graphicx,lscape,enumitem,dsfont,bm,url,subfigure}
\usepackage{epsfig,amstext,xspace}
\usepackage{algorithm,algorithmic,comment}

\usepackage{thmtools,thm-restate}

\usepackage{algorithm,algorithmic}
\usepackage{auxiliary}
\usepackage[utf8]{inputenc} 
\usepackage[T1]{fontenc}    
\usepackage{hyperref}       
\usepackage{url}            
\usepackage{booktabs}       
\usepackage{amsfonts}       
\usepackage{nicefrac}       
\usepackage{microtype}      

\usepackage{color}              
\usepackage[suppress]{color-edits}
\addauthor{tl}{cyan}
\addauthor{sv}{blue}

\begin{document}
\title{Competitive caching with machine learned advice}

\author{
 Thodoris Lykouris\thanks{Microsoft Research, \texttt{thlykour@microsoft.com}. Work supported under NSF grant CCF-1563714. This work was done while the author was a Ph.D. student at Cornell University and started while the author was interning at Google.} \and 
 Sergei Vassilvitskii\thanks{Google Research, \texttt{sergeiv@google.com}.}}

\date{First version: February 2018\\Current version: August 2020%
\footnote{Preliminary version of the paper appeared at the 35th International Conference on Machine Learning (ICML 2018). The current version improves the presentation of the suggested framework (Section~\ref{sec:model}), provides a more clear discussion on how it can be more broadly applied, and fixes some more minor presentation issues in other sections.}}

\maketitle

\begin{abstract}

 Traditional online algorithms encapsulate decision making under uncertainty and give ways to hedge against all possible future events, while guaranteeing a nearly optimal solution, as compared to an offline optimum.  On the other hand, machine learning algorithms are in the business of extrapolating patterns found in the data to predict the future, and usually come with strong guarantees on the expected generalization error. 

In this work we develop a framework for augmenting online algorithms with a machine learned oracle to achieve competitive ratios that provably improve upon unconditional worst case lower bounds when the oracle has low error. Our approach treats the oracle as a complete black box, and is not dependent on its inner workings, or the exact distribution of its errors. 

We apply this framework to the traditional caching problem—creating an eviction strategy
for a cache of size $k$. We demonstrate that naively following the oracle's recommendations
may lead to very poor performance, even when the average error is quite low. Instead we show how to modify the Marker algorithm to take into account the oracle's predictions, and prove that this combined approach achieves a competitive ratio that both decreases as the oracle's error decreases, and is always capped by $O(\log k)$, which can be
achieved without any oracle input. We complement our results with an empirical evaluation of our algorithm on real world datasets, and show that it performs well empirically even when using simple off-the-shelf predictions.

\end{abstract}

\addtocounter{page}{-1}
\thispagestyle{empty}
\newpage

\section{Introduction}
\label{sec:intro}

Despite the success and prevalence of machine learned systems across many application domains, there are still a lot of hurdles that one needs to overcome to deploy an ML system in  practice~\cite{MLTechDebt}. As these systems are rarely perfect, a key challenge is dealing with errors that inevitably arise.

There are many reasons that learned systems may exhibit errors when deployed. First, most of them are trained to be good {\em on average}, minimizing some expected loss. In doing so, the system may  
invest its efforts in reducing the error on the majority of inputs, at the 
expense of increased error on a handful of outliers. Another problem is
that generalization error guarantees only apply when the train and test examples are drawn from the same distribution. If this assumption is violated, either due to distribution drift or adversarial examples~\cite{FoolingNNs}, the machine learned predictions may be very far from the truth. In all cases, any system backed by machine learning needs to be robust enough to handle occasional errors. 

While machine learning is in the business of predicting the unknown, online algorithms provide guidance on how to act without \emph{any} knowledge of future inputs. These powerful methods show how to hedge decisions so that, regardless of what the future holds, the online algorithm performs nearly as well as the optimal offline algorithm. However these guarantees come at a cost: since they protect against the worst case, online algorithms may be overly cautious, which translates to high competitive ratios  even for seemingly simple problems. 

In this work we ask:

\begin{center}\emph{What if the online algorithm is equipped with a machine learned oracle? How can one use this oracle to combine the predictive power of machine learning with the robustness of online algorithms?}
\end{center}

We focus on a prototypical example of this area: the online paging, or \emph{caching} problem. In this setting, a series of requests arrives one at a time to a server equipped with a small amount of memory. Upon processing a request, the server places the answer in the memory (in case an identical request comes in the future). Since the local memory has limited size, the server must decide which of the current elements to evict. It is well known that if the local memory or \emph{cache} has size $k$, then any deterministic algorithm incurs competitive ratio $\Omega(k)$. However, an $O(k)$ bound can be also achieved by almost any reasonable strategy, thus this metric fails to distinguish between algorithms that perform well in practice and those that perform poorly. The competitive ratio of the best randomized algorithm is $\Theta(\log k)$ which, despite its elegant analysis, is much higher than what is observed on real inputs. 

In contrast, we show how to use machine learned predictions to achieve a competitive ratio of $2 + O( \min( \sqrt{\nicefrac{\eta}{\opt}}, \log k)$, when using a predictor with total prediction error of $\eta$, where $\opt$ is the value of the offline optimal solution --- see Section \ref{ssec:predictive_marker} for a precise statement of results. Thus when the predictions are accurate (small $\eta$), our approach circumvents the worst case lower bounds. On the other hand, even when the oracle is inaccurate (large $\eta$), the performance degrades gracefully to almost match the worst case bound.

\subsection{Our contribution}
The conceptual contribution of the paper lies in formalizing the interplay between machine learning and competitive analysis by introducing a general framework (Section \ref{sec:general_omla}), and a set of desiderata for online algorithms that use a machine learned oracle.

We look for approaches that:
\begin{itemize}
    \item Make \emph{minimal} assumptions on the machine learned oracle. Specifically, since most machine learning guarantees are on the expected performance, our results are parametric as a function of the error of the oracle, $\eta$, and not the distribution of the error. 
    \item Are \emph{robust}: a better oracle (one with lower $\eta$) results in a smaller competitive ratio
    \item Are worst-case \emph{competitive}: no matter the performance of the oracle on the particular instance, the algorithm behaves comparably to the best online algorithm for the problem.
\end{itemize}

We instantiate the general framework to the online caching problem, specifying the prediction made by the oracle, and presenting an algorithm that uses that prediction effectively (Section \ref{ssec:predictive_marker}). Along the way we show that algorithmic innovation is necessary: simply following the recommendations of the predictor may lead to poor performance, even when the average error is small (Section \ref{ssec:blind_oracle_lower_bounds_robustness}). Instead, we adapt the Marker algorithm~\cite{Fiat1991} to carefully incorporate the feedback of the predictor. The resulting approach, which we call {\em Predictive Marker} has guarantees that capture the best of both worlds: the algorithm performs better as the error of the oracle decreases, but performs nearly as well as the best online algorithm in the worst case. 

Our analysis generalizes to multiple loss functions (such as absolute loss and squared loss). This freedom in the loss function with the black-box access to the oracle allows our results to be strengthened with future progress in machine learning and reduces the task of designing better algorithms to the task of finding better  predictors.

We complement our theoretical findings with empirical results (Section~\ref{sec:empirical}). We test the performance of our algorithm on public data using off-the-shelf machine learning models.  We compare the performance to the Least Recently Used (LRU) algorithm, which serves as the gold standard, the original Marker algorithm, as well as directly using the predictor. In all cases, the Predictive Marker algorithm outperforms known approaches. 

Before moving to the main technical content, we provide a simple example that highlights the main concepts of this work. 

\subsection{Example: Faster Binary Search}
\label{sec:example}
Consider the classical binary search problem. Given a sorted array $A$ on $n$ elements and a query element $q$, the goal is to either find the index of $q$ in the array, or state that it is not in the set. The textbook method is binary search: compare the value of $q$ to that of the middle element of $A$, and recurse on the correct half of the array. After $O(\log n)$ probes, the method either finds $q$ or returns. 

Instead of applying binary search, one can train a classifier, $h$, to predict the position of $q$ in the array. (Although this may appear to be overly complex, Kraska et.al~\cite{PredictiveIndexing} empirically demonstrate the advantages of such a method.) How to use such a classifier? A simple approach is to first probe the location at $h(q)$; if $q$ is not found there, we immediately know whether it is smaller or larger. Suppose $q$ is larger than the element in $A[h(q)]$ and the array is sorted in increasing order. We probe elements at $h(q) + 2, h(q) + 4, h(q) + 8$, and so on, until we find an element larger than $q$ (or we hit the end of the array). Then we simply apply binary search on the interval that's guaranteed to contain $q$. 

What is the cost of such an approach? Let $t(q)$ be the true position of $q$ in the array (or the position of the largest element smaller than $q$ if it is not in the set). The absolute loss of the classifier on $q$ is then $\epsilon_q = |h(q) - t(q)|$.
On the other hand, the cost of running the above algorithm starting at $h(q)$ is at most $2(\log{ \abs{h(q)-t(q)}}) = 2 \log \epsilon_q$. 

If the queries $q$ come from a distribution, then the expected cost of the algorithm is:
$$2 \mathbb{E}_q \Big[\log \left(|h(q) - t(q)|\right)\Big] \leq 2 \log \mathbb{E}_q \Big[|h(q) - t(q)|\Big] = 2 \log \mathbb{E}_q[\epsilon_q],$$

where the inequality follows by Jensen's inequality. This gives a trade-off between the performance of the algorithm and the absolute loss of the predictor. Moreover, since $\epsilon_q$ is trivially bounded by $n$, this shows that even relatively weak classifiers (those with average error of $\sqrt{n}$) this can lead to an improvement in asymptotic performance.

\subsection{Related work}
\label{ssec:related}
Our work builds upon the foundational work on competitive analysis and online algorithms; for a great introduction see the book by Borodin and El-Yaniv~\cite{BorodinBook}. Specifically, we look at the standard caching problem, see, for example, ~\cite{Motwani1995}. While many variants of caching have been studied over the years, our main starting point will be the Marker algorithm by Fiat et al. ~\cite{Fiat1991}. 

As we mentioned earlier, competitive analysis fails to distinguish between algorithms that perform well in practice, and those that perform well only in theory. Several fixes have been proposed to address these concerns, ranging from resource augmentation, where the online algorithm has a larger cache than the offline optimum ~\cite{ResourceAugmentation}, to models of real-world inputs that restrict the inputs analyzed by the algorithm, for example, insisting on locality of reference ~\cite{LocalityOfReference}, or the more general Working Set model ~\cite{Denning1968}. 

The idea of making assumptions on the nature of the input to prove better bounds is common in the literature. The most popular of these is that the data arrive in a random order. This is a critical assumption in the secretary problem, and, more generally, in other streaming algorithms, see for instance the survey by McGregor~\cite{GraphStream}. While the assumption leads to algorithms with useful insight into the structure of the problem, it rarely holds true, and is often hard to verify.

Another common assumption on the structure of the input gives rise to Smoothed Analysis, introduced in a pioneering work by Spielman and Teng~\cite{SpielmanTeng2004} explaining the practical efficiency of the Simplex method. This approach assumes that any worst case instance is perturbed slightly before being passed to the algorithm;  the idea is that this perturbation may be due to measurement error, or some other noise inherent in the data. The goal then is to show that the worst case inputs are brittle, and do not survive addition of random noise. Since its introduction this method has been used to explain the unusual effectiveness of many practical algorithms such as ICP~\cite{ArthurVassilvitskii2006}, Lloyd's method~\cite{Arthur+2011}, and local search~\cite{Englert+2016}, in the face of exponential worst case bounds.

The prior work that is closest in spirit to ours looks for algorithms that optimistically assume that the input has a certain structure, but also have worst case guarantees when that fails to be the case. One such assumption is that the data are coming from a stochastic distribution and was studied in the context of online matching \cite{soda/MirrokniGZ12} and bandit learning \cite{BubeckS12}; both of these works provide improved guarantees if the input is stochastic but retain the worst-case guarantees otherwise. 
Subsequent work has provided a graceful decay in performance when the input is mostly stochastic (analogous to our robustness property) both in the context online matching \cite{esfandiari2015online} and bandit learning \cite{LykourisMiPa18}. In a related note, Ailon et al.~\cite{Ailon+2011} consider ``self-improving'' algorithms that effectively learn the input distribution, and adapt to be nearly optimal in that domain. Contrasting to these works, our approach utilizes a different structure in the data: the fact that the sequence can be predicted.

Our work is not the first to use predictions to enhance guarantees in online decision-making. The ability to predict something about the input has also used been used in online learning by Rakhlin and Sridharan \cite{RakhlinSr13_predictable} where losses of next round are predicted and the guarantees scale with how erroneous these precitions are. Our focus is on competitive analysis approaches where requests affect the state of the system; as a result, a single misprediction can have long-lasting effect on the system. With respect to using predictions in competitive analysis, another approach was suggested by Mahdian et al.~\cite{MahdianNazerzadehSaberi2012}, who assume the existence of an optimistic, highly competitive, algorithm, and then provide a meta algorithm with a competitive ratio that interpolates between that of the worst-case algorithm and that of the optimistic one.
This work is most similar to our approach, but it ignores two key challenges that we face: (i) identifying predictions that can lead to (near) offline optimality, and (ii)  developing algorithms that use these predictions effectively in a robust way.
The work of Mahdian et al. \cite{MahdianNazerzadehSaberi2012} starts directly from the point where such an ``optimistic'' algorithm is available, and combines it with a ``good in the worst-case''  algorithm in a black-box manner. This has similarities to the approaches we discuss in Section~\ref{ssec:black_box} and  Remark~\ref{rem:fiat_reference} but does not answer how to develop the optimistic algorithm. As we show in the paper, developing such  algorithms may be non-trivial even when the predictions are relatively good.

In other words, we do not assume anything about the data, or the availability of good algorithms that work in restricted settings.
Rather, we use the predictor to implicitly classify instances into ``easy'' and ``hard'' depending on their predictability. The ``easy'' instances are those on which the latest machine learning technology, be it perceptrons, decision trees, SVMs, Deep Neural Networks, GANs, LSTMs, or whatever else may come in the future, has small error. On these instances our goal is to take advantage of the predictions, and obtain low competitive ratios. (Importantly, our approach is completely agnostic to the inner workings of the predictor and treats it as a black box.) The ``hard'' instances, are those where the prediction quality is poor, and we have to rely more on classical competitive analysis to obtain good results. 

A previous line of work has also considered the benefit of enhancing online algorithms with oracle advice (see \cite{Boyar2016_survey} for a recent survey). This setting assumes access to an infallible oracle and studies the amount of information that is needed to achieve desired competitive ratio guarantees. Our work differs in two major regards. First, we do not assume that the oracle is perfect, as that is rarely the case in machine learning scenarios. Second, we study the trade-off between oracle error and the competitive ratio, rather than focusing on the number of perfect predictions necessary. 

Another avenue of research close to our setting asks what happens if the algorithm cannot view the whole input, but must rely on a sample of the input to make its choices.  Introduced in the seminal work of Cole and Roughgarden~\cite{ColeRoughgarden2014}, this notion of Learning from Samples, can be viewed as first designing a good prediction function, $h$, and then using it in the algorithms. Indeed, some of the follow up work~\cite{MorgensternR16, Balkanski}  proves tight bounds on precisely how many samples are necessary to achieve good approximation guarantees. In contrast, we assume that the online algorithm is given access to a machine learned predictor, but does not know any details of its inner workings---we know neither the average performance of the predictor, nor the distribution of the errors. 

Very recently, two papers explored domains similar to ours. Medina and Vassilvitskii ~\cite{MedinaVassilvitskii2017} showed how to use a machine learned oracle to  optimize revenue in repeated posted price auctions. Their work has a mix of offline calculations and online predictions and focuses on the specific problem of revenue optimization. Kraska et al.~\cite{PredictiveIndexing} demonstrated empirically that introducing machine learned components to classical algorithms (in their case index lookups) can result in significant speed and storage gains. Unlike this work, their results are experimental, and they do not provide trade-offs on the performance of their approach vis-\`a-vis the error of the machine learned predictor.

Finally, since the publication of the original paper, learning augmented algorithms has emerged as a rich area. Subsequently to our work, researchers have studied how to incorporate machine learned predictions in other settings such as ski rental \cite{purohit2018improving,GollapudiPanigrahi19}, 
scheduling \cite{purohit2018improving,angelopoulos2019online,LattanziLavastidaMoseleyVassilvitskii20}, bin packing \cite{angelopoulos2019online}, bloom filters ~\cite{Mitzenmacher_bloom, partitionedbloom2020}, queueing \cite{queueing2020},
streaming algorithms ~\cite{HsuIndykKatabiVakilian19}, weighted paging \cite{JiangPanigrahiSun20}, and page migration \cite{indyk2020online}. While many of these focus on improving competitive ratios, some of them explore other performance metrics, such as space complexity~\cite{HsuIndykKatabiVakilian19, Mitzenmacher_bloom, partitionedbloom2020}. With respect to the unweighted caching problem we consider, subsequent work has also provided refined guarantees under our prediction model \cite{Rohatgi20,Wei20} or alternate prediction models \cite{AntoniadisCoesterEliasPolakSimon20}.

\section{Online Algorithms with Machine Learned Advice}
\label{sec:general_omla}
In this section, we introduce a general framework for combining online algorithms with machine learning predictions, which we term \emph{Online with Machine Learned Advice} framework (OMLA). Before introducing the framework, we review some basic notions from machine learning and online algorithms. 

\subsection{Preliminaries}
\paragraph{Machine learning basics.}\label{ssec:ml_basics}
We are given a feature space $\mathcal{X}$, describing the salient characteristics of each item and a set of labels $\mathcal{Y}$. An example is a pair $(x,y)$, where $x\in \mathcal{X}$ describes the specific features of the example, and $y\in \mathcal{Y}$ gives the corresponding label. In the binary search example, $x$ can be thought as the query element $q$ searched and $y$ as its true position $t(x)$.

A hypothesis is a mapping $h:\mathcal{X}\rightarrow \mathcal{Y}$ and can be probabilistic in which case the output on $x\in \mathcal{X}$ is some probabilistically chosen $y\in \mathcal{Y}$. In binary search, $h(x)$ corresponds to the predicted position of the query.

To measure the performance of a hypothesis, we first define a loss function $\ell:\mathcal{Y}\times \mathcal{Y}\rightarrow \mathbb{R}^{\geq 0}$.
When the labels lie in a metric space, we define absolute loss $\ell_{1}(y,\hat{y})=\abs{y-\hat{y}}$, squared loss $\ell_{2}(y,\hat{y})=(y-\hat{y})^2$, and, more generally,  classification loss $\ell_{c}(y, \hat{y}) = \mathbf{1}_{y\not=\hat{y}}$. 

\paragraph{Competitive analysis.}
To obtain worst-case guarantees for an online algorithm (that must make decisions as each element arrives), we compare its performance to that of an offline optimum (that has the benefit of hindsight). Let $\sigma$ be the input sequence of elements for a particular online decision-making problem, $\cost_A(\sigma)$ be the cost incurred by an online algorithm $\mathcal{A}$ on this input, and $\opt(\sigma)$ be the cost incurred by the optimal offline algorithm. Then algorithm $\mathcal{A}$ has competitive ratio $\textsc{cr}$ if for all sequences $\sigma$,
$$\cost_
{\mathcal{A}}(\sigma)\leq
\textsc{cr}\cdot \opt(\sigma).$$ 
If the algorithm $\mathcal{A}$ is randomized then $\cost_{\mathcal{A}}(\sigma)$ corresponds to the expected cost of the algorithm in input $\sigma$ where the expectation is taken over the randomness of the algorithm.

\paragraph{The Caching Problem.} The caching (or online paging) problem considers a system with two levels of memory: a slow memory of size $m$ and a fast memory of size $k$. A caching algorithm is faced with a sequence of requests for elements. If the requested element is in the fast memory, a {\em cache hit} occurs and the algorithm can satisfy the request at no cost. If the requested element
is not in the fast memory, a {\em cache miss} occurs, the algorithm fetches the element from the slow memory, and places it in the fast memory before satisfying the request. If the fast memory is full, then one of the elements
must be evicted. The eviction strategy forms the core of the  problem. The goal is to find an eviction policy that results in the fewest number of cache misses. 

It is well known that the optimal offline algorithm at time $t$ evicts the element from the cache that will arrive the furthest in the future; this is typically referred to in the literature as B\'el\'ady's optimal replacement paging algorithm \cite{Belady1966}. On the other hand, without the benefit of foresight, any deterministic caching algorithm achieves a competitive ratio of $\Omega(k)$, and any randomized caching algorithm achieves a competitive ratio of $\Omega(\log k)$~\cite{Motwani1995}. 

\subsection{OMLA Definition.}
\label{sec:model}
To define our framework in generality, we consider a general problem setting associated with a general prediction model and then explain how both can be instantiated in the context of caching.

In traditional online algorithms, there is an universe $\mathcal{Z}$ and elements $z_i\in\mathcal{Z}$ arrive one at a time for rounds $i=1,2,\ldots$. The problem $\Pi$ specifies the optimization problem at hand, along with the required constraints and any necessary parameters. 
For example, in the problem of caching studied in this paper, $\Pi_{\textsc{caching}}=\textsc{Caching}(n,k)$, is parametrized by the number of requests $n$ and the cache size $k$.

\paragraph{Augmenting online algorithms with machine learned predictors.} In our framework, we assume that the requested elements are augmented with a feature space $\mathcal{X}$ (discussed below). We refer to the resulting feature-augmented elements as \emph{items} and denote the item of the $i$-th request by $\sigma_i$. An input $\sigma\in\Pi$ corresponds to a sequence of items: $\sigma=(\sigma_1,\sigma_2,\ldots)$. For the problem of caching $\Pi_{\textsc{caching}}=\textsc{Caching}(n,k)$, the item sequence $\sigma$ has length $n$.

Each item is associated with a particular element by $z(\sigma_i)\in\mathcal{Z}$ as well as a feature $x(\sigma_i)\in\mathcal{X}$. The features capture any information that may be available to the machine learning algorithm to help provide meaningful predictions. In caching, these may include information about the sequence prior to the request, time patterns associated to the particular request, or any other information. We note that even for caching, items are more general than their associated element: two items with the same element are \emph{not} necessarily the same as their corresponding features may differ.

\paragraph{Prediction model.} The prediction model $\mathcal{H}$ prescribes a label space $\mathcal{Y}$; the $i$-th item
has label $y(\sigma_i)\in\mathcal{Y}$. This label space can be viewed as the information needed to solve the task (approximately) optimally. As we discuss in the end of Section~\ref{sec:model}, deciding on a particular label space is far from trivial and it often involves trade-offs between learnability and accuracy.

Given a prediction model $\mathcal{H}$ determining a label space $\mathcal{Y}$, a machine learned predictor $h\in\mathcal{H}$ maps features $x\in\mathcal{X}$ to predicted labels $h(x)\in\mathcal{Y}$. In particular, for item
$\sigma_i$, the predictor $h$ returns a predicted label $h(x(\sigma_i))$. To ease notation we
denote this by $h(\sigma_i)$. Here we assume that this mapping from features to labels is deterministic; our results extend to randomized mappings by applications of Jensen's inequality (see Section~\ref{sec:randomized_predictors}).

\paragraph{Loss functions and error of predictors.} To evaluate the performance of a predictor on a particular input, we consider a loss function $\ell$. Similar to the prediction model, selecting a loss function involves trade-offs between learnability of the predictor and resulting performance guarantees; we elaborate on these trade-offs in the end of Section~\ref{sec:model}. For a given loss function $\ell$, problem $\Pi$, and prediction model $\mathcal{H}$, the performance of the predictor $h\in\mathcal{H}$ on input $\sigma\in\Pi$ is evaluated by its error $\eta_\ell(h,\sigma)$. In full generality, this error can depend on the whole input in complicated ways.

For the caching problem, the prediction model we consider predicts the subsequent time that a requested element will get requested again. In this case, a natural loss function such as absolute or squared loss decomposes the error across items.
In later sections, we focus on such loss functions throughout this paper and therefore can express the error as:
$$
\eta_{\ell}(h,\sigma)=\sum_i \ell\prn*{y(\sigma_i), h(\sigma_i)}. $$
Instantiated with the absolute loss function, the error of the predictor is $\eta_{\ell_1}(h,\sigma)=\sum_i |y(\sigma_i) - h(\sigma_i)|$. We will use $\eta_1(h,\sigma)$ as a shorthand for this absolute loss.

We note that this decomposition across items may not be possible. For example, edit distance does not decompose across items but needs to be evaluated with respect to the whole instance. The general framework we define extends to such non-decomposable loss functions but the above restriction lets us better describe our results and draws more direct connection with classical machine learning notions.

We now summarize the general concepts of our framework in the following definition.

\begin{definition}\label{defn:omla}
\vspace{0.1in}
The {\em Online with Machine Learned Advice} (OMLA) framework is defined with respect to a) a problem $\Pi$, b) a prediction model $\mathcal{H}$ determining a feature space $\mathcal{X}$ and a label space $\mathcal{Y}$, and c) a loss function $\ell$. An instance consists of:
\begin{itemize}
    \item An input $\sigma\in\Pi$ consisting of items
    $\sigma_i$ arriving online, each with features $x(\sigma_i)\in\mathcal{X}$ and label $y(\sigma_i)\in\mathcal{Y}$.
    \item A predictor $h :\mathcal{X} \rightarrow \mathcal{Y}$ that predicts a label $h(\sigma_i)$ for each $x(\sigma_i) \in \mathcal{X}$.
    \item The error of predictor $h$ at sequence $\sigma$ w.r.t. loss $\ell$, $\eta_{\ell}(h,\sigma)$.
\end{itemize}
\end{definition}

Our goal is to create online algorithms that, when augmented with a predictor $h$, can use its advice to achieve 
an improved competitive ratio. To evaluate how well an algorithm $\mathcal{A}$ performs with respect to this task, we extend the definition of competitive ratio to be a function of the predictor's error. We first define the set of predictors that are sufficiently accurate. 

\begin{definition}\label{defn:epsilon_accurate}\vspace{0.1in}
For a fixed optimization problem $\Pi$, let $\opt_\Pi(\sigma)$ denote the value of the optimal solution on the input $\sigma$. Consider a prediction model $\mathcal{H}$.
A predictor $h\in \mathcal{H}$ is {\em $\epsilon$-accurate} with respect to a loss function $\ell$ for $\Pi$ if for any $\sigma\in\Pi$:
$$\eta_{\ell,\mathcal{H},\Pi}(h, \sigma) \leq \epsilon \cdot \opt_\Pi (\sigma).$$
We will use $\mathcal{H}_{\ell,\mathcal{H},\Pi}(\epsilon)$ to denote the class of $\epsilon$-accurate predictors. 
\end{definition}

At first glance, it may appear unnatural to tie the error of the prediction to the value of the optimal solution. However, our goal is to have a definition that is invariant to simple padding arguments. For instance, consider a sequence $\sigma' = \sigma \sigma$, which concatenates two copies of an input $\sigma$.~\footnote{In order for both instances to be equally sized and therefore be inputs of the same problem $\Pi$, we can think of padding the end of the first instance with the same dummy request.} It is clear that the prediction error of any predictor doubles, but this is not due to the predictor suddenly being worse. One could instead normalize the prediction error by the length of the sequence, but in many problems, including caching, one can artificially increase the length of the sequence without impacting the value of the optimum solution, or the impact of predictions. Normalizing by the value of the optimum addresses both of these problems. 

Call an algorithm $\mathcal{A}$ {\em $\epsilon$-assisted} if it has access to an $\epsilon$-accurate predictor. The competitive ratio of an  $\epsilon$-assisted algorithm is itself a function of $\epsilon$ and may also depend on parameters specified by $\Pi$ such as the cache size $k$ or the number of elements $n$.

\begin{definition}\label{defn:predictor_augmented_cr}
\vspace{0.1in} 
For a fixed optimization problem $\Pi$ and a prediction model $\mathcal{H}$, let $\textsc{inputCR}_{\mathcal{A},\mathcal{H},\Pi}(h,\sigma)$
be the competitive ratio of algorithm $\mathcal{A}$ that
uses a predictor $h\in\mathcal{H}$ when applied on an input $\sigma\in\Pi$.
The \emph{competitive ratio} of an $\epsilon$-assisted
algorithm  $\mathcal{A}$ for problem $\Pi$ with respect to loss function $\ell$ and prediction model $\mathcal{H}$ is:
$$
\textsc{cr}_{\mathcal{A},\ell,\mathcal{H},\Pi}(\epsilon)=\max_{\sigma\in \Pi,h\in \mathcal{H}_{\ell,\mathcal{H},\Pi}(\epsilon)}\textsc{inputCR}_{\mathcal{A},\mathcal{H},\Pi}(h,\sigma).
$$
\end{definition}
We now define the desiderata that we wish our algorithm to satisfy. 
We would like our algorithm to perform as well as the offline optimum when the predictor is perfect, degrade gracefully with the error of the predictor, and perform as well as the best online algorithm regardless of the error of the predictor. We define these properties formally for the performance of an algorithm $\mathcal{A}$ a particular loss function $\ell$, prediction model $\mathcal{H}$, and problem $\Pi$. 

\begin{definition} 
\vspace{0.1in}
$\mathcal{A}$ is $\beta$-{\em consistent} if $\textsc{cr}_{\mathcal{A},\ell,\mathcal{H},\Pi}(0) = \beta$.  
\end{definition}

\begin{definition}
\vspace{0.1in}
$\mathcal{A}$ is $\alpha$-{\em robust} for a function $\alpha(\cdot)$, if $\textsc{cr}_{\mathcal{A},\ell,\mathcal{H},\Pi}(\epsilon) = O(\alpha(\epsilon))$. 
\end{definition}

\begin{definition}
\vspace{0.1in}
$\mathcal{A}$ is $\gamma$-{\em competitive} if $\textsc{cr}_{\mathcal{A},\ell,\mathcal{H},\Pi}(\epsilon)\leq \gamma$ for all values of $\epsilon$.
\end{definition}

Our goal is to find algorithms that simultaneously optimize the aforementioned three properties.
First, they are ideally $1$-consistent: recovering the optimal solution when the predictor is perfect. This is not necessarily feasible for multiple reasons. From a computational side, the underlying problem may be NP-hard. Moreover, achieving any notion of robustness may inevitably be at odds with exact consistency. As a result, we are satisfied with $\beta$-consistency for some small constant $\beta$. Second, they are $\alpha(\cdot)$-robust for a slow growing function $\alpha$: seamlessly handling errors in the predictor. This function depends on the exact prediction model and the way that the loss is defined with respect to it. As discussed below, different prediction models and loss functions may well lead to different robustness guarantees while also achieve different levels of learnability. Finally, they are  worst-case competitive: they perform as well as the best online algorithms even when the predictor's error is high. As hinted before, all competitive ratios can be functions of the problem dimensions inherent in $\Pi$; for example, in caching, the worst-case performance $\gamma$ needs to depend on the cache size $k$. Ideally, the consistency and robustness quantities $\beta$ and $\alpha(\epsilon)$ (for small $\epsilon>0$) do not display such dependence on these problem dimensions.

\paragraph{Discussion on the OMLA framework.} 
For the caching problems predictions and loss functions as decomposable per element, but one can also define predictions with respect to different parts of the instance. For example, subsequent works used strong lookahead for weighted paging
\cite{JiangPanigrahiSun20} and learned weights for scheduling \cite{LattanziLavastidaMoseleyVassilvitskii20} -- both of these prediction models are not per-element. Similarly, loss functions can be computed with respect to the complete instance.
Per-item predictions, however, have a stronger connection to classical machine learning terminology.

Next, thus far we have disregarded the question of where the predictor comes from and how learnable it is. This is an important question and has been elegantly discussed in multiple contexts such as revenue maximization \cite{ColeRoughgarden2014}. In general, the decision on both the prediction model $\mathcal{H}$ and the loss function $\ell$ needs to take into account the learnability question and better understanding of the exact trade-offs is a major open direction of our work. Subsequent work sheds further light on the learnability question in the context of our framework \cite{AnandGePanigrahi20}.

Finally, although we define our framework with respect to competitive analysis, predictions can be useful to augment online algorithm design with respect to other metrics such as space complexity~\cite{HsuIndykKatabiVakilian19, Mitzenmacher_bloom, partitionedbloom2020} and our framework can be easily extended to capture such performance gains.

\subsection{Caching with ML Advice}

In order to instantiate the framework to the caching problem, we need to specify the items
of the input sequence $\sigma$, the prediction model $\mathcal{H}$ (and thereby the label space $\mathcal{Y}$), as well as the loss function $\ell$. Each item
corresponds to one request $\sigma_i$; the latter is associated with an element $z(\sigma_i)\in\mathcal{Z}$  and features $x(\sigma_i)\in\mathcal{X}$ that encapsulate any information available to the machine learning algorithm. The element space $\mathcal{Z}$ consists of the $m$ elements of the caching problem defined in Section~\ref{ssec:ml_basics}. The exact choice of the feature space $\mathcal{X}$ is orthogonal to our setting, though of course richer features typically lead to smaller errors. The input sequence $\sigma=(\sigma_1,\sigma_2,\ldots)$ of the requested items is assumed to be fixed in advance and is oblivious to the realized randomness of the algorithm but unknown to the algorithm.

The main design choice of the prediction model is the question of what to predict which is captured in our framework by the choice of the label space. For caching problems, a natural candidate is predicting the next time a particular element is going to appear. It is well known~\cite{Belady1966} that when such predictions are perfect, the online algorithm can recover the best offline optimum. 
Formally, the label space $\mathcal{Y}$ we consider is a set of positions in the sequence, $\mathcal{Y} = \mathbb{N}^+$. Given a sequence $\sigma$, the label of the $i$-th element is $y(\sigma_i) = \min_{t>i}\crl*{t: x(\sigma_t) = x(\sigma_i)}$. If the element is never seen again, we set $y(\sigma_i) = n + 1$. Note that $y(\sigma_i) $ is completely determined by the sequence $\sigma$. We use ${h}(\sigma_i)$ to denote the outcome of the prediction on an element with features $x(\sigma_i)$. Note that the feature is not only a function of the element identity $z(\sigma_i)$; when an element reappears, its features may be drastically different.

In what follows, we fix the problem $\Pi=\textsc{Caching}(n,k)$ to a caching problem with $n$ requests and cache size $k$ and the prediction model $\mathcal{H}$ to be about the next appearance of a requested element. We consider a variety of loss functions (discussed in detail in Section~\ref{ssec:analysis_pred_marker}) that capture for example absolute and squared loss functions. To ease notation, we therefore drop any notational dependence on the prediction model $\mathcal{H}$ and the problem $\Pi$ as both are fixed throughout the rest of the paper, but keep the dependence on the loss function $\ell$.

\section{Main result: Predictive Marker}
\label{sec:result}
In this section, we describe the main result: an algorithm that satisfies the three desiderata of the previous section. Before describing our algorithm, we show that combining the predictions with ideas from competitive analysis is to a large extent essential; blindly evicting the element that is predicted the furthest in the future by the predictor (or simple modifications of this idea) can result in poor performance both with respect to robustness and competitiveness.

\subsection{Blindly following the predictor is not sufficient}\label{ssec:blind_oracle_lower_bounds_robustness}
\paragraph{Evicting element predicted the furthest in the future.} An immediate way to use the predictor is to treat its output as truth and optimize based on the whole predicted sequence. This corresponds to
B\'el\'ady rule that evicts the element predicted to appear the furthest in the future.  We refer to this algorithm as algorithm $\mathcal{B}$ (as it follows the B\'el\'ady rule). Since this rule achieves offline optimality, this approach is consistent, i.e. if the predictor is perfect, this algorithm is ex-post optimal. Unfortunately this approach does not have similarly nice performance with respect to the other two desiderata. With respect to robustness, the degradation with the average error of the predictor is far from the best possible, while a completely unreliable predictor leads to unbounded competitive ratios, far from the ones of the best online algorithm.

\begin{proposition}\label{thm:blind_belady}
\vspace{0.1in}
Consider the caching problem $\Pi$ with $n$ requests and cache size $k$, the prediction model $\mathcal{H}$ that predicts the next arrival of a requested element and the absolute loss function $\ell_1$.
The competitive ratio of $\epsilon$-assisted algorithm $\mathcal{B}$ is $\textsc{cr}_{\mathcal{\mathcal{B}},\ell_1}(\epsilon)=\Omega(\epsilon)$.
\end{proposition}
The implication is that when the error of the predictor is much worse than the offline optimum, the competitive ratio becomes unbounded. With respect to robustness, the rate of decay is far from optimal as we will see in Section~\ref{ssec:analysis_pred_marker}.

\begin{proof}[Proof of Proposition~\ref{thm:blind_belady}]
We will show that for every $\epsilon$, there exist a sequence $\sigma$ and a predictor $h$ such that the absolute error $\eta_{1}(h,\sigma)\leq \epsilon \cdot \opt$ while the competitive ratio of algorithm $\mathcal{B}$ is $\frac{\epsilon-1}{2}$. For ease of presentation, assume that $\epsilon>3$. 
Consider a cache of size $k=2$ and three elements $a,b,c$; the initial configuration of cache is $a,c$. The sequence consists of repetitions of the following sequence with $\epsilon$ requests per repetition.
The actual sequence is
$a\underbrace{bcbc\ldots bc}_{\epsilon-1} a\underbrace{bcbc\ldots bc}_{\epsilon-1} \ldots$
($a$ appears once and then $bc$ appears $(\epsilon-1)/2$ times). 

In any repetition, the predictor accurately predicts the arrival time of all elements apart from two: i) when element $a$ arrives, it predicts that it will arrive again two steps after the current time (instead of in the first step of the next repetition) and ii) when $b$ arrives for the last time in one repetition, it predicts it to arrive again in the fourth position of the next repetition (instead of the second). As a result, the absolute error of the predictor is $\epsilon$ ($\epsilon-2$ error in the $a$-misprediction and $2$ error in the $b$-misprediction).
The optimal solution has two evictions
per repetition (one to bring $a$ in the cache and one to directly evict it afterwards). Instead, the algorithm never evicts $a$ as it is predicted to arrive much earlier than all other elements, and therefore has $\epsilon-1$ cache misses. This means that the competitive ratio of this algorithm is $\Omega(\eta_1(h,\sigma)/\opt(\sigma))$ which completes the proof.
\end{proof}

\paragraph{Evicting elements with proven wrong predictions.} The problem in the above algorithm is that algorithm $\mathcal{B}$ keeps too much faith in predictions that have been already proven to be wrong (as the corresponding elements are predicted to arrive in the past). It is tempting to ``fix'' the issue by evicting any element whose predicted time has passed, and evict again the element predicted the furthest in the future if no
such element exists. We call this algorithm $\mathcal{W}$ as it takes care of wrong predictions. Formally, let $h(j,t)$ denote the last prediction about $z_j$ at or prior to time~$t$. At time $t$ algorithm $\mathcal{W}$ evicts an arbitrary element
from the set  $S_t=\crl*{j:h(j,t)<t}$ if $S_t\neq \emptyset$ and $\arg\max_{z_i\in \text{Cache(t)}} h(i,t)$ otherwise. We show that algorithm $\mathcal{W}$ has similarly bad performance guarantees.

\begin{proposition}
\vspace{0.1in}
Consider the caching problem $\Pi$ with $n$ requests and cache size $k$, the prediction model $\mathcal{H}$ that predicts the next arrival of a requested element and the absolute loss function $\ell_1$. The competitive ratio of $\epsilon$-assisted algorithm $\mathcal{W}$ is $\textsc{cr}_{\mathcal{W},\ell_1}(\epsilon)=\Omega(\epsilon)$.
\end{proposition}
\begin{proof}
Consider a cache of size $k=3$ and four elements $a,b,c,d$; the initial configuration of cache is $a,b,c$ and then $d$ arrives. The actual sequence consists of repetitions of the following sequence with $(\epsilon/2)+1$ requests per repetition (for ease of presentation, assume that $\epsilon>6)$. 
The actual sequence $\sigma$ is $d\underbrace{abcabc\ldots abc}_{\nicefrac{\epsilon}{2}}d\underbrace{abcabc\ldots abc}_{\nicefrac{\epsilon}{2}}\ldots$. 

In any repetition, the predictor $h$ accurately predicts the arrival time of element $d$ but always makes mistake in elements $a,b,c$ by predicting them to arrive two time steps earlier. As a result, the absolute error of the predictor is $\epsilon$ (error of $2$ for any of the appearances of $a,b,c$).
The optimal solution has two evictions
per repetition (one to bring element $d$ and one to evict it afterwards). Instead the algorithm always evicts elements $a,b,c$ because they are predicted earlier than their actual arrival and are therefore evicted as ``wrong'' predictions. This means that the competitive ratio of this algorithm is also $\Omega(\eta_{1}(h,\sigma)/\opt(\sigma))$ which completes the proof.
\end{proof}
The latter issue can be again fixed by further modifications of the algorithm but these simple examples demonstrate that, unless taken into account, mispredictions can cause significant inefficiency in the performance of the algorithms.

\paragraph{Beyond blindly trusting the predictor.} The common problem in both examples is that there is an element that should be removed but the algorithm is tricked into keeping it in the cache. To deal with this in practice,  most popular heuristics such as LRU (Least Recently Used) and FIFO (First In First Out) avoid evicting recent elements when some elements have been dormant for a long time. This tends to utilize nice locality properties leading to strong empirical performance (especially for LRU). However, such heuristics impose a strict eviction policy which leads to weak performance guarantees. Moreover, incorporating additional information provided by the predictor becomes complicated. 

Competitive analysis has also built on the idea of evicting dormant elements via developing algorithms with stronger theoretical guarantees such as Marker. In the next subsection, we show how we can incorporate predictions in the Marker algorithm to enhance its performance when the predictions are good while retaining the worst-case guarantees. Interestingly, via our framework, we can provide improved guarantees for the aforementioned heuristics such as LRU, improving their worst-case guarantees while retaining their practical performance (see Section~\ref{ssec:practical_predictive}).

\subsection{Predictive Marker Algorithm}
\label{ssec:predictive_marker}

We now present our main technical contribution, 
a prediction-based adaptation of the Marker algorithm \cite{Fiat1991}. This $\epsilon$-assisted algorithm gets a competitive ratio of $2 \cdot \min(\tledit{1+\sqrt{5\epsilon}}
, 2H_k)$ where $H_k=1+\nicefrac{1}{2}+\dots+\nicefrac{1}{k}$ denotes the $k$-th Harmonic number.

\paragraph{Classic Marker algorithm.} We begin by recalling the Marker algorithm and the analysis of its performance. The algorithm runs in phases. At the beginning of each phase, all elements are unmarked. When an element arrives and is already in the cache, the element is marked. If it is not in the cache, a {\em random unmarked} element is evicted, the newly arrived element is placed in the cache and is marked.  Once all elements are marked and a new cache miss occurs, the phase ends and we unmark all of the elements. 

For the purposes of analysis, an element is called \emph{clean} in phase $r$ if it appears during phase $r$, but does not appear during phase $r-1$. In contrast, elements that also appeared in the previous phase are called \emph{stale}. The marker algorithm has \tledit{a} competitive ratio of $2H_k - 1$ and the analysis is tight~\cite{AchlioptasCN00}. We use a slightly simpler analysis that achieves competitive ratio of $2H_k$ below. 

The crux of the upper bound lies in two claims. The first relates the performance of the optimal offline algorithm to the total number of clean elements $Q$ across all phases. 

\begin{claim}[\cite{Fiat1991}]\label{lem:opt_clean_marker}
Let $Q$ be the number of clean elements. Then the optimal algorithm suffers at least $\nicefrac{Q}{2}$ cache misses.
\end{claim}
The second comes from bounding the performance of the algorithm as a function of the number of clean elements. 
\begin{claim}[\cite{Fiat1991}]\label{lem:random_marker}
Let $Q$ be the number of clean elements. Then the expected number of cache misses of the marker algorithm is $Q\cdot H_k$.
\end{claim}

\paragraph{Predictive Marker.}\label{ssec:algorithm}
The algorithm of ~\cite{Fiat1991} is part of a larger family of {\em marking} algorithms; informally, these algorithms
never evict marked elements when there are unmarked elements present. Any algorithm in this family has a worst-case competitive ratio of $k$. Therefore pairing predictions with a marking algorithm would avoid the pathological examples we saw previously.  

A natural approach is to use predictions for tie-breaking, specifically evicting the element whose predicted next appearance time is furthest in the future. When the predictor is perfect (and has zero error), the stale elements never result in cache misses, and therefore, by Claim~\ref{lem:opt_clean_marker}, the algorithm has a competitive ratio of $2$. On the other hand, by using the Marker algorithm and not blindly trusting the oracle, we can guarantee a worst-case competitive ratio of $k$.

We extend this direction to further reduce the worst-case competitive ratio  to $O(H_k)$. To achieve this, we combine the prediction-based tie-breaking rule with the random tie-breaking rule. Suppose an element $e$ is evicted during the phase. We construct a blame graph to understand the reason why $e$ is evicted; this may happen for two distinct reasons. First, $e$ may have been evicted when a clean element $c$ arrived; in this case, we create a new node $c$ which can be thought as the start of a distinct chain of nodes. Alternatively, it may have been evicted because a stale element $s$ arrived ($s$ was previously evicted in the same phase); in this case, we add a directed edge from $e$ to $s$.
Note that the graph is always a set of chains (paths). The total length of the chains represents the total number of evictions incurred by the algorithm during the phase, whereas the number of distinct chains represents the number of clean elements. We call the lead element in every chain {\em representative} and denote it by $\omega(r, c)$, where $r$ is the index of the phase and $c$ the index of the chain in the phase.

Our modification is simple -- when a stale element arrives, it evicts a new element in a prediction-based manner if the chain containing it has length less than $H_k$. Otherwise it evicts a random unmarked element. 
Looking ahead to the analysis, this switch to uniform evictions results in at most $H_k$ additional elements added to any chain during the course of the phase. This guarantees that the competitive ratio is at most $O(H_k)$ in expectation; we make the argument formal in Theorem~\ref{thm:main}.

The key to the analysis is the fact that the chains are disjoint, thus the interactions between evictions can be decomposed cleanly. We give a formal description of the algorithm in Algorithm \ref{alg:clean_chains}. For simplicity, we drop dependence on $\sigma$ from the notation.

\begin{algorithm}[!h]
\caption{Predictive Marker}
\begin{algorithmic}[1]\label{alg:clean_chains}
\REQUIRE Cache $\mathcal{C}$ of size $k$ initially empty ($\mathcal{C}\leftarrow \emptyset$).
\STATE Initialize phase counter $r\leftarrow 1$, unmark all elements ($\mathcal{M}\leftarrow \emptyset$), and set round $i\leftarrow 1$.
\STATE Initialize clean element counter $q_r\leftarrow 0$ and tracking set $\mathcal{S}\leftarrow \emptyset$.
\STATE Element $z_i$ arrives, and the predictor gives a prediction $h_i$. Save prediction $p(z_i)\leftarrow h_i$.
\IF {$z_i$
results in cache hit or the cache is not full ($z_i\in \mathcal{C}$ or $\abs{\mathcal{C}}<k$) } 
    \STATE Add to cache $C\leftarrow C\cup\{z_i\}$ without evicting any element and go to step 26 
\ENDIF
    \IF {the cache is full and all cache elements are marked ($\abs{\mathcal{M}}=k$) }
    \STATE Increase phase ($r\leftarrow r+1$), initialize clean counter ($q_r\leftarrow 0$), save current cache
    ($\mathcal{C}\rightarrow \mathcal{S}$) as the set of elements that are possibly stale in the new phase, and unmark elements ($\mathcal{M}\leftarrow \emptyset$).
    \ENDIF
    \IF{$z_i$ is a clean element ($z_i\notin\mathcal{S}$)}
    \STATE Increase number of clean elements: $q_r\leftarrow q_r+1$. 
    \STATE Initialize size of new clean chain: $n(r,q_r)\leftarrow 1$.
    \STATE Select to evict unmarked element with highest predicted time: $e=\argmax_{z\in\mathcal{C}-\mathcal{M}}p(z)$.
    \ENDIF
    \IF {$z_i$ is a stale element ($z_i\in\mathcal{S}$)}
    \STATE It is the representative of some clean chain. Let $c$ be this clean chain: $z_i=\omega(r,c)$. 
    \STATE Increase length of the clean chain $n(r,c)\leftarrow n(r,c)+1$.
    \IF {$n(r,c)\leq H_k$}
    \STATE Select to evict unmarked element with highest predicted time: $e=\argmax_{z\in\mathcal{C}-\mathcal{M}}p(z)$.
    \ELSE
     \STATE Select to evict a random unmarked element $e\in \mathcal{C}-\mathcal{M}$.
    \ENDIF
    \STATE Update cache by evicting $e$: $\mathcal{C}\leftarrow \mathcal{C}\cup \{z_i\}-\{e\}$.
    \STATE Set $e$ as representative for the chain: $\omega(r,c)\leftarrow e$. 
    \ENDIF
\STATE Mark incoming element ($\mathcal{M}\leftarrow \mathcal{M}\cup\{z_i\}$), increase round ($i\leftarrow i+1$), and go to step 3.
\end{algorithmic}
\end{algorithm}

\subsection{Analysis}\label{ssec:analysis_pred_marker}
In order to analyze the performance of the proposed algorithm, we begin with a technical definition that captures how slowly a loss function $\ell$ can grow.
\begin{definition}\label{defn:spread}
Let $\mathcal{A}_T$ be the set of all the sequences $A_T = a_1, a_2, \ldots, a_T$ of increasing integers of length $T$, that is $a_1 < a_2 < \ldots < a_T$, and $\mathcal{B}_T$ be the set of all sequences
$B_T = b_1, b_2, \ldots, b_T$ of non-increasing reals of length $T$, $b_1 \geq b_2 \geq \ldots \geq b_T$. For a fixed loss function $\ell$, we define its {\em spread} $S_\ell: \mathbb{N}^+ \to \mathbb{R}^+$ as:
$$S_{\ell} (m) = \min\{T: \forall A_T\in\mathcal{A}_T,B_T\in\mathcal{B}_T: \ell(A_T, B_T) \geq m\}$$
\end{definition}
The spread captures the length of a subsequence that can be predicted in completely reverse order as a function of the error of the predictor with respect to loss function $\ell$. We note that the sequence $B_T$ is assumed to be over reals instead of integers as it corresponds to the outcome of the machine learned predictor and we do not want to unnecessarily restrict the output of this predictor.

The following Lemma instantiates the spread for loss metrics we consider and is proved in the Appendix ~\ref{app:spread_common_losses}. 
\begin{lemma}\label{lem:spread_common_losses}
\vspace{0.1in}
For absolute loss, $\ell_1(A,B) = \sum_i |a_i - b_i|$, the spread of $\ell_1$ is $S_{\ell_1}(m) \leq \sqrt{5m}$. \\
For squared loss, $\ell_2(A, B) = \sum_i (a_i - b_i)^2$, the spread of $\ell_2$ is $S_{\ell_2}(m) \leq \sqrt[3]{14m}.$ 
\end{lemma}

We now prove the main theorem of the paper. \begin{theorem}
\label{thm:main}
\vspace{0.1in}
Consider the caching problem $\Pi$ with $n$ requests and cache size $k$, the prediction model $\mathcal{H}$ that predicts the next arrival of a requested element and any
loss function $\ell$ with spread bounded by $S_\ell$ for some function $S_{\ell}$ that is concave in its argument. Then the competitive ratio of $\epsilon$-assisted Predictive Marker \textsc{PM} is bounded by:
$$\textsc{cr}_{\mathcal{\textsc{PM},\ell}}(\epsilon)\leq 2\cdot \min\left(1 + 2S_{\ell}\left(\epsilon
\right), 2H_k\right).$$
\end{theorem}
To prove this theorem, we first introduce an analogue of Claim \ref{lem:random_marker}, which decomposes the total cost into that incurred by each of the chains individually. 

To aid in our analysis, we consider the following marking algorithm, which we call \textsc{SM}  (Special Marking). Initially we simply evict an arbitrary unmarked element. At some point, the adversary designates an arbitrary element not in the cache as {\em special}. For the rest of the  phase, upon a cache miss, if the arriving element is  special, the algorithm evicts a {\em random} unmarked element and designates the evicted element as special. If the arriving element is not special, the algorithm proceeds as before, evicting an arbitrary unmarked element.

\begin{lemma}\label{lem:chain_random_eviction}
\vspace{0.1in}
Using algorithm \textsc{SM}, in expectation at most $H_k$ special elements cause cache misses per phase. 
\end{lemma}
\begin{proof}
Since we use a marking algorithm, the set of elements that are in the cache at the end of each phase is determined by the element sequence $(z_1,z_2,\ldots)$ and is independent of the particular eviction rule among unmarked elements. Fix a phase that begins at time $\tau$. Let $E$ be the set of $k$ distinct elements that arrive in this phase. Note that the arrival of the $k+1$st distinct elements starts a new phase. 

Consider the  time $\tau^{\star}$ that an element is designated special and assume that, at this time, there are $i^{\star}$  special elements. At this point, we define $A\subseteq E$ to be the subset of the initial elements that are unmarked and in the cache; we refer to this set as the candidate special set as they are the only ones that can subsequently get designated as special; the set's initial cardinality is $i^{\star}$. This set is shrinking over time as elements are getting marked or evicted from the cache. Order the elements by the time of their first requst in this phase. 

We now bound the probability of the event $\mathcal{E}_i$ that an element becomes special when it is the $i$-th last element in $A$ (based on the ordering by first arrival). By principle of deferred decisions, we consider the first time that, upon request of a special element, it evicts one of the last $i$ elements in the active set. If this never happens then the event $\mathcal{E}_i$ never occurs.
Otherwise, observe that we select the element to evict uniformly at random, and there exists at least one element in the cache that never appears before the end of the phase. Second, if at any point an element $j$ among the $i-1$ elements in the active set becomes special, the $i$-th element can no longer become special as, at the time that $j$ is requested, $i$ is already marked. The above imply that the probability of the event $\mathcal{E}_i$ is at most:
\begin{equation}
    \label{eqn:prob}
Pr[\mathcal{E}_i] \leq \frac{1}{i+1}.
\end{equation}
Therefore, given Equation \eqref{eqn:prob}, we can bound the expected number of misses caused by special elements as:
$$1 + \sum_{i=1}^{k-1} \frac{1}{i+1} = H_k,$$
where the first term is due to the first special element and the second term is due to events $\mathcal{E}_1$ through $\mathcal{E}_{k-1}$.
\end{proof}
We now provide the lemma that lies in the heart of our robustness property.

\begin{lemma}\label{lem:robustness_chain}
\vspace{0.1in}
For any loss metric $\ell$, any phase $r$, the expected length of any chain is at most $1+\mathcal{S}_{\ell}(\eta_{r,c})$ where $\eta_{r,c}$ is the cumulative error of the predictor on the elements in the chain and $\mathcal{S}_{\ell}$ is the spread of the loss metric.
\end{lemma}
\begin{proof}
The clean element that initiates the clean chain evicts one of the unmarked elements upon arrival. Since it does so based on the B\'el\'ady rule, it evicts the element $s_1$ that is predicted to reappear the latest in the future. If the predictor is perfect, this element will never appear in this phase. If, on the other hand, $s_1$ comes back (is a stale element) let $s_2$ be the element it evicts, which is predicted to arrive the furthest among the current unmarked elements. 

Suppose there are $m$ such evictions: $s_1,s_2,\dots, s_m$. The elements were predicted to arrive in reverse order of their evictions. This is because elements $s_j$ for $j>i$ were unmarked and in the cache when element $s_i$ got evicted; therefore $s_i$ was predicted to arrive later. However, the actual arrival order is the reverse. If $\eta_{r,c}$ is the total error of these elements, setting the actual arriving times as the sequence $A_T$ and the predicted ones as the sequence $B_T$ in the definition of spread (Definition \ref{defn:spread}), it means that $m\leq S_{\ell}(\eta_{r,c})$.
\end{proof}

Combining the two above lemmas, we can obtain a bound on the expected length of any chain.

\begin{lemma}\label{lem:chain_length_bound}
\vspace{0.1in}
For any loss metric $\ell$, any phase $r$, the expected length of any chain is at most $\min(1+2\mathcal{S}_{\ell}(\eta_{r,c}), 2\log k)$ where $\eta_{r,c}$ is the cumulative error of the predictor on the elements in the chain and $\mathcal{S}_{\ell}$ is the spread of the loss metric.
\end{lemma}
\begin{proof}
The proof follows from combining the two above lemmas. By Lemma~\ref{lem:chain_random_eviction}, if the chain switches to random evictions, it incurs another $H_k$ cache misses in expectation after the switch point (and its length increases by the same amount), capping in expectation the total length by $2H_k \leq 2 \log k$. 
If the chain does not switch to random evictions, it has B\'el\'ady evictions and, by Lemma~\ref{lem:robustness_chain}, it incurs at most $\mathcal{S}_{\ell}(\eta_{r,c})$ misses from stale elements. To ensure that the $2\log k$ term dominates the bound when $S_{\ell}(\eta_{r,c})\geq \log k$, we multiply $S_{\ell}(\eta_{r,c})$ by a factor of $2$ in the first term.
\end{proof}

\begin{proof}[Proof of Theorem \ref{thm:main}]
Consider an input $\sigma\in \Pi$ determining the request sequence.
Let $Q$ be the number of clean elements (and therefore also chains). Any cache miss corresponds to a particular eviction in one clean chain; there are no cache misses not charged to a chain by construction. By Lemma \ref{lem:chain_length_bound}, we can bound the evictions from the clean chain $c$ of the $r$-th phase in expectation by $\min(1+2\cdot S_{\ell}(\eta_{r,c}), 2\log k)$. Since both $S_\ell$ and the minimum operator are concave functions, the way to maximize the length of chains is to apportion the total error, $\eta$, equally across all of the chains. Thus  for a given error $\eta$ and number $Q$ of clean chains, the competitive ratio is maximized when the error in each chain is $\eta_{r,c}=\eta/Q$ each. The total number of stale elements is therefore in expectation at most:
$Q\cdot \min(2\cdot S_{\ell}(\nicefrac{\eta}{Q}),2H_k)$.
By Claim \ref{lem:opt_clean_marker}, it holds that $\nicefrac{Q}{2} \leq \opt(\sigma)$, implying the result since 
$\opt(\sigma) \leq Q$. 
\end{proof}
We now specialize the results for the absolute and squared losses.

\begin{corollary}
\label{cor:l1}
\vspace{0.1in}
The competitive ratio of $\epsilon$-assisted Predictive Marker with respect to the absolute loss metric $\ell_1$ is bounded by $\textsc{cr}_{\mathcal{\mathcal{PM}},\ell_1}(\epsilon)\leq \min\prn*{2+2\cdot\sqrt{5\epsilon},4H_k}$.
\end{corollary}

\begin{corollary}
\label{cor:squared}
\vspace{0.1in}
The competitive ratio of $\epsilon$-assisted Predictive Marker with respect to the absolute loss metric $\ell_2$ is bounded by
$\textsc{cr}_{\mathcal{\mathcal{PM}},\ell_1}(\epsilon)\leq \min\prn*{2+2\cdot\sqrt[3]{{14\epsilon}},4H_k}$.
\end{corollary}

\subsection{Tightness of analysis}\label{ssec:rate_robustness}
\paragraph{Robustness rate of Predictive Marker.} We show that our analysis is tight: any marking algorithm that uses the predictor in a deterministic way cannot achieve an improved guarantee with respect to robustness. 

\begin{theorem}
\label{thm:lower}\vspace{0.1in}
Any deterministic $\epsilon$-assisted marking algorithm $\mathcal{A}$, that only uses the predictor in tie-breaking among unmarked elements in a deterministic fashion, has a competitive ratio of $\textsc{cr}_{\mathcal{A,\ell}}(\epsilon)=\Omega\prn*{\min\prn*{\mathcal{S}_{\ell}\prn*{\epsilon},k}}$.
\end{theorem}
\begin{proof}
Consider a cache of size $k$ with $k+1$ elements and any $\epsilon$ such that $S_{\ell}(\epsilon)<k$. We will construct an instance that exhibits the above lower bound. Since $\mathcal{A}$ uses marking, we can decompose its analysis into phases. Let $\sigma$ be the request sequence, and assume that we do not have any repetition of an element
inside the phase; as a result the $i$-th element of phase $r$ corresponds to element $\sigma_{(r-1)k+i}$. 

Suppose the predictor is always accurate on elements 2 through $k-S_{\ell}(\eps)+1$ in each phase.

For the last $S_{\ell}(\eps)-1$ elements of phase $r$ as well as the first element of the of the next phase, the elements are predicted to come again at the beginning of the subsequent phase, at time $t=rk+1$. Since the algorithm is deterministic, we order the elements so that their evictions are in reverse order of their arriving time. By the definition of spread, the error of the predictor in these elements is exactly $\epsilon$ and the algorithm incurs a cache miss in each of them. On the other hand, the offline optimum has only $1$ miss per phase, which concludes the proof. 
\end{proof}

\paragraph{On the rate of robustness in caching.} Theorem~\ref{thm:lower} establishes that the analysis of Predictive Marker is tight with respect to the rate of robustness, and suggests that algorithms that use the predictor in a deterministic manner may suffer from similar rates. 
However, a natural question that comes up is whether a better rate can be achieved using the predictor in a randomized way. We conjecture that a rate of $\log(1+\sqrt{
\epsilon})$ with respect to the absolute loss is possible, similar to the exponential improvement randomized schemes obtain over the deterministic guarantees of $k$ with respect to worst-case competitiveness. In subsequent work, Rohatgi \cite{Rohatgi20} made significant progress towards identifying the correct rate by proving refined upper and lower bounds.

\subsection{Randomized predictors}\label{sec:randomized_predictors}
We now remove the assumption that the predictor $h$ is deterministic and extend the definition of $\epsilon$-accurate predictors (Definition~\ref{defn:epsilon_accurate}) to hold in expectation. The randomness may either come in how the inputs are generated or in the predictions of $h$.
\begin{definition}\vspace{0.1in}
For a fixed optimization problem $\Pi$, let $\opt_{\Pi}(\sigma)$ denote the value of the optimal solution on input $\sigma$. Assume that the predictor is probabilistic and therefore the error of the predictor at $\sigma$ is a random variable $\eta_{\ell}(h,\sigma)$. Taking expectation over the randomness of the precictor, we say that a predictor $h$ is $\epsilon$-accurate in expectation for $\Pi$ if:
$$
\En\brk*{\eta_{\ell}(h,\sigma)}\leq \eps\cdot \opt_{\Pi}(\sigma).
$$
Similarly an algorithm is $\epsilon$-assisted if it has access to an $\epsilon$-accurate predictor in expectation.
\end{definition}

Analogously to the previous part, we can show:

\begin{theorem}
\label{thm:main_random}
\vspace{0.1in}
Consider any loss function $\ell$ with spread bounded by $S_\ell$ for some function $S_{\ell}$ that is concave in its argument. Then the competitive ratio of $\epsilon$-assisted in expectation Predictive Marker \textsc{PM} is bounded by:
$$\textsc{cr}_{PM,\ell}(\epsilon)\leq 2\cdot \min\left(1 + 2S_{\ell}\left(\epsilon\right), 2H_k\right).$$
\end{theorem}
\begin{proof}
For ease of notation assume that the outcomes of the predictors are finite. For each of these potential realizations, we can bound the performance of the algorithm  by Theorem \ref{thm:main}. The proof then follows by applying an additional Jensen's inequality on all the possible realizations due to the concavity of the spread and the min operator.
\end{proof}

\section{Discussion and extensions
}\label{ssec:remarks}
Thus far we have shown how to use an $\epsilon$-accurate predictor to get a caching algorithm with an $O(\sqrt{\epsilon})$ competitive ratio for the absolute loss metric. We now provide a deeper discussion of the main results. 
In Section \ref{ssec:tradeoff}, we give a finer trade-off of competitiveness and robustness. We then discuss some traits that limit the impact of the errors of the predictors in Section \ref{ssec:practical_predictive}.Subsequently, we show that common heuristic approaches, such as LRU, can be expressed as predictors in our framework. This allows us to combine their predictive power with robust guarantees when they fail. Finally, in Section \ref{ssec:black_box}, we provide a black-box way to combine robust and competitive approaches.

\subsection{Robustness vs competitiveness trade-offs.}
\label{ssec:tradeoff}
One of the free parameters in Algorithm \ref{alg:clean_chains} is the length of the chain when the algorithm switches from following the predictor to random evictions. If the switch occurs after chains grow to $\gamma H_k$ in length, this provides a trade-off between competitiveness and robustness. 
\begin{theorem}
\label{thm:tradeoff}
\vspace{0.1in}
Suppose that, for $\gamma>0$, the algorithm uses $\gamma H_k$ as switching point (line 18 in Algorithm~\ref{alg:clean_chains}); denote this algorithm by $\mathcal{PM}(\gamma)$. Let a 
loss function $\ell$ with spread bounded by $S_\ell$ for some function $S_{\ell}$ that is concave in its argument. Then
the  competitive ratio of $\epsilon$-assisted  $\mathcal{PM}(
\gamma)$ is bounded by:
$$\textsc{cr}_{\mathcal{PM}(\gamma),\ell}(\epsilon)\leq 2\cdot \min\left(1 + \frac{1+\gamma}{\gamma}S_{\ell}\left(\epsilon\right), \gamma H_k, k\right).$$
\end{theorem}
\begin{proof}
The proof follows the proof of Theorem~\ref{thm:main} but slightly changes the Lemma~\ref{lem:chain_random_eviction} to account for the new switching point. In particular, with respect to the second term, the expected length of each clean chain is at most $H_k$ after the switching point, and, at most $\gamma H_k$ before the switching point by construction. 

With respect to the robustness term, the length of each clean chain before the switch is bounded by the spread of the metric on this subsequence. Since the total length is in expectation at most $(1+\gamma)/\gamma$ higher, we need to adjust the first term accordingly.

Finally, the length of its clean chain is at most $k$ regardless of the tie-breaking since we are using marking which provides the last term.
\end{proof}

Let us reflect on the above guarantee. When $\gamma\rightarrow 0$ then the algorithm is more conservative (switching to random evictions earlier); this reduces the worst-case competitive ratio but at the cost of abandoning the predictor unless it is extremely accurate. On the other hand, setting $\gamma$ very high makes the algorithm trust the predictor more, reducing the competitive ratio when the predictor is accurate at the expense of a worst guarantee when the predictor is unreliable.

\subsection{Practical traits of Predictive Marker}\label{ssec:practical_predictive}

\paragraph{Locality.} The guarantee in Theorem \ref{thm:main} bounds the competitive ratio as a function of the quality of the prediction. One potential concern is that if the predictions have of a small number of very large errors, then the applicability of Predictive Marker may be quite limited. 

Here we show that this is not the case. Due to the phase-based nature of the analysis, the algorithm essentially ``resets'' at the end of every phase, and therefore the errors incurred one phase do not carry over to the next. Moreover, the competitive ratio in every phase is bounded by $O(H_k)$. 

Formally, for any sequence $\sigma$, we can define phases that consist of exactly $k$ distinct elements. Let $\textsc{cl}(r,\sigma)$ be the number of clean elements in phase $r$ of sequence $\sigma$, and let $\eta_{\ell, r}(h, \sigma)$ denote the error of predictor $h$ restricted only to  elements occurring in phase $r$. 

\begin{theorem}\label{thm:locality}
\vspace{0.1in}
Consider a loss function $\ell$ with spread $S_\ell$. If $S_{\ell}$ is concave, the competitive ratio of Predictive Marker $PM$ at sequence $\sigma$ when assisted by a predictor $h$ is at most:
$$\textsc{cr}_{\mathcal{\mathcal{PM},\ell}}\leq \frac{\sum_r \textsc{cl}(r,\sigma)\cdot \min\prn*{1+2S_{\ell}(\eta_{\ell,r}(h,\sigma), 2H_k}}{\sum _r \textsc{cl}(r,\sigma)}$$
\end{theorem}
\begin{proof}
The proof follows directly from Lemma~\ref{lem:chain_length_bound} and applying Jensen's inequality only within the chains of the phase (instead of also across phases as we did in Theorem~\ref{thm:main}).
\end{proof}
This theorem illustrates a nice property of our algorithm. If the predictor $h$ is really bad for a period of time (i.e. its errors are localized) then the clean chains of the corresponding phases will contribute the second term (the logarithmic worst-case guarantee) but the other phases will provide enhanced performance utilizing the predictor's advice. In this way, the algorithm adapts to the quality of the predictions, and bad errors do not propagate beyond the end of a phase. 
This quality is very useful in caching where most patterns are generally well predicted but there may be some unforeseen sequences.

\paragraph{Robustifying LRU.}
Another practical property of our algorithm is that it can seamlessly incorporate heuristics that are known to perform well in practice. In particular, the popular Least Recently Used (LRU) algorithm can be expressed within the Predictive Marker framework. Consider the following predictor, $h$: when an element $\sigma_i$ arrives at time $i$, the LRU predictor predicts next arrival time $h(\sigma_i)=-i$. 

Note that, by doing so, at any point of time, among the elements that are in the cache, the element that is predicted the furthest in the future is exactly the one that has appeared the least recently. Also note that any marked element needs to have arrived later than any unmarked element. As a result, if we never switched to random evictions (or had $k$ in the RHS of line 18 in Algorithm~\ref{alg:clean_chains}), the Predictive Marker algorithm assisted with the LRU predictor is exactly the LRU algorithm. 

The nice thing that comes from this observation is that we can robustify the analysis of LRU. LRU, and its variants like LRU(2), tend to have very good empirical performance as using the recency of requests is a good predictor about how future requests will arise. However, the worst-case guarantee of LRU is unfortunately $\Theta(k)$ since it is a deterministic algorithm. By expressing LRU as a predictor in the Predictive Marker framework and using a switching point of $H_k$ for each clean chain, we exploit most of this predictive power while also guaranteeing a logarithmic worst-case bound on it. 

\subsection{Combining robustness and competitiveness in a black-box manner}\label{ssec:black_box}
In the previous section, we showed how we can slightly modify a classical competitive algorithm to ensure that it satisfies nice consistency and robustness properties when given access to a good predictor, while retaining the worst-case competitiveness guarantees otherwise. In this part, we show that, in fact, achieving the requirements individually is enough. In particular, we show a black-box way to combine an algorithm that is robust and one that is worst-case competitive.  This reduction leads to a slightly worse bound, but shows that proving the robustness property (i.e. a graceful degradation with the error of the predictor) is theoretically sufficient to augment an existing worst-case competitive algorithm.

\begin{theorem}\label{thm:black_box}
\vspace{0.1in}
For the caching problem, let $A$ be an $\alpha$-robust algorithm and $B$ a $\gamma$-competitive algorithm. We can then create a black-box algorithm $ALG$ that is both $9\alpha$-robust and $9\gamma$-competitive.
\end{theorem}
\begin{proof}
We proceed by simulating $A$ and $B$ in parallel on the dataset, and maintaining the cache state and the number of misses incurred by each. Our algorithm switches between following the strategy of $A$ and the strategy of $B$. Let $c_t(A)$ and $c_t(B)$ denote the cost (number of misses) of $A$ and $B$ up to time $t$. Without loss of generality, let $ALG$ begin by following strategy of $A$; it will do so until a time $t$ where $c_t(A) = 2 \cdot c_t(B)$. At this point $ALG$ switches to following the eviction strategy of $B$, doing so until the simulated cost of $B$ is double that of $A$: a time $t'$ with $c_{t'}(B) = 2\cdot c_{t'}(A)$. At this point it switches back to following eviction strategy of $A$, and so on. When $ALG$ switches from $A$ to $B$, the elements that $A$ has in cache may not be the same as those that $B$ has in the cache. In this case, it needs to reconcile the two. However, this can be done lazily (at the cost of an extra cache miss for every element that needs to be reconciled). To prove the bound on the performance of the algorithm, we next show that $c_t(ALG) \leq 9 \cdot \min(c_t(A), c_t(B))$ for all $t$. We decompose the cost incurred by $ALG$ into that due to following the different algorithms, which we denote by $f_t(ALG)$, and that due to reconciling caches, $r_t(ALG)$. 

We prove a bound on the following cost $f_t$ by induction on the number of switches. Without loss of generality, suppose that at time $t$, $ALG$ switched from $A$ to $B$, and at time $t'$ it switches from $B$ back to $A$. By induction, suppose that $f_t(ALG) \leq 3 \min(c_t(A), c_t(B)) = 3 c_t(B)$, where the equality follows since $ALG$ switched from $A$ to $B$ at time $t$. In both cases, assume that caches are instantly reconciled. Then:
\begin{align*}
    f_{t'}(ALG) &= f_{t}(ALG) + (c_{t'}(B) - c_{t}(B)) \\
    &= f_{t}(ALG) + 2c_{t'}(A) - \nicefrac{1}{2}c_{t}(A) \\
    &\leq 3c_{t}(B) + 2(c_{t'}(A) -c_{t}(A)) + \nicefrac{3}{2} \cdot c_{t}(A) \\
    &=3c_{t}(A) + 2(c_{t'}(A) - c_{t}(A)) \\
    &\leq 3c_{t'}(A) \\
    &=3 \min(c_{t'}(A), c_{t'}(B))
\end{align*}
What is left is to bound the following cost for the time since the last switch. Let $s$ denote the time of the last switch and, assume without loss of generality that it was done from $A$ to $B$. Let  $s'$ denote the last time step. By the previous set of inequalities (changing the second equation to inequality) and the fact that the algorithm never switched back to $A$ after $s$, it holds that $f_{s'}(ALG)\leq 3c_{s'}(A)\leq 6\min(c_{s'}(A),c_{s'}(B))$.

To bound the reconciliation cost, assume the switch at time $t$ is from $A$ to $B$. We charge the reconciliation of each element in $B \setminus A$ to the cache miss when the element was last evicted by $A$. Therefore the overall reconciliation cost is bounded by $r_t(ALG)\leq c_{t}(A) + c_{t}(B) \leq 3\min(c_{t} (A), c_{t} (B)$.
\end{proof}
Observe that the above construction can extend beyond caching and applies to any setting where we can bound the cost that the algorithm needs to incur to reconcile the states of the robust and the worst-case competitive algorithm. In particular, this occurs in the more general $k$-server problem.

\begin{remark}\label{rem:fiat_reference}
\vspace{0.1in}
The above construction is similar to that of Fiat et al.~\cite{FiatRabaniRavid94} who showed how to combine multiple competitive algorithms. In subsequent work, Antoniadis et al.~\cite{AntoniadisCoesterEliasPolakSimon20} relied on a similar construction to provide results for metrical task systems under a different prediction model.
\end{remark}

\section{Experiments}
\label{sec:empirical}

In this section we evaluate our approach on real world datasets, empirically demonstrate its dependence on the errors in the oracle, and compare it to standard baselines.

\paragraph{Datasets and Metrics}
We consider two datasets taken from different domains to demonstrate the wide applicability of our approach. 
\begin{itemize}
\item \texttt{BK} is data extracted from BrightKite, a now defunct social network. We consider sequences of checkins, and extract the top $100$ users with the longest non-trivial check in sequences---those where the optimum policy would have at least $50$ misses. This dataset is publicly available at ~\cite{Cho11,Brightkite}. Each of the user sequences represents an instance of the caching problem.  
\item \texttt{Citi} is data extracted from CitiBike, a popular bike sharing platform operating in New York City. We consider citi bike trip histories, and extract stations corresponding to starting points of each trip. We create 12 sequences, one for each month of 2017 for the New York City dataset. We consider only the first 25,000 events in each file. This data is publicly available at ~\cite{CitiData}.  
\end{itemize}

We give some additional statistics about each datasets in Table \ref{tab:stats}. 

\begin{table}[t!]
\centering
\begin{tabular}{c|c|c|c}
     Dataset & Num Sequences & Sequence Length & Unique Elements  \\
     \hline
     \texttt{BK} & 100 & 2,101 & 67-- 800 \\
     \texttt{Citi} & 24 & 25,000 & 593 -- 719 \\
\end{tabular}
\caption{Number of sequences; sequence length;  min and max number of elements for each dataset.}
\label{tab:stats}
\end{table}

Our main metric for evaluation will be the {\em competitive ratio} of the algorithm, defined as the number of misses incurred by the particular strategy divided by the optimum number of misses.

\paragraph{Predictions}
We run experiments with both synthetic predictions to showcase the sensitivity of our methods to learning errors, and with preditions using an off the shelf classifier, published previously~\cite{Anderson2014}.

\begin{itemize}
    \item {\bf Synthetic Predictions.}  For each element, we first compute the true next arrival time, $y(t)$, setting it to $n+1$ if it does not appear in the future. To simulate the performance of an ML system, we set $h(t) = y(t) + \epsilon$, where $\epsilon$ is drawn i.i.d. from a lognormal distribution with mean parameter $0$ and standard deviation $\sigma$. We chose the lognormal distribution of errors to showcase the effect of rare but large failures of the learning algorithm. Finally, observe that since we only compare the relative predicted times for each method, adding a bias term to the predictor would not change the results. 
    \item {\bf PLECO Predictions.} In their work Anderson et al.~\cite{Anderson2014} developed a simple framework to model repeat consumption, and published the parameters of their PLECO (Power Law with Exponential Cut Off) model for the BrightKite dataset. While their work focused on predicting the relative probabilities of each element (re)appearing in the subsequent time step, we modify it to predict the next time an element will appear. Specifically, we set $h(t) = t + 1/{p(t)}$, where $p(t)$ represents the probability that element that appeared at time $t$ will re-appear at time $t+1$. 
\end{itemize}

\paragraph{Algorithms}
We consider multiple algorithms for evaluation. 
\begin{itemize}
    \item {\bf LRU} is the Least Recently Used policy that is wildly successful in practice. 
    \item {\bf Marker} is the classical Marker algorithm due to Fiat et al.~\cite{Fiat1991}.
    \item {\bf PredictiveMarker} is the algorithm we develop in this work. We set the switching cost to $k$, and therefore never switch to random evictions.
    \item {\bf Blind Oracle} is the algorithm $\mathcal{B}$ described in Section \ref{ssec:blind_oracle_lower_bounds_robustness}, which evicts the element predicted to appear furthest in the future.  
\end{itemize}

\subsection{Results}

\begin{figure}[t!]
\centering
\includegraphics[scale=0.28]{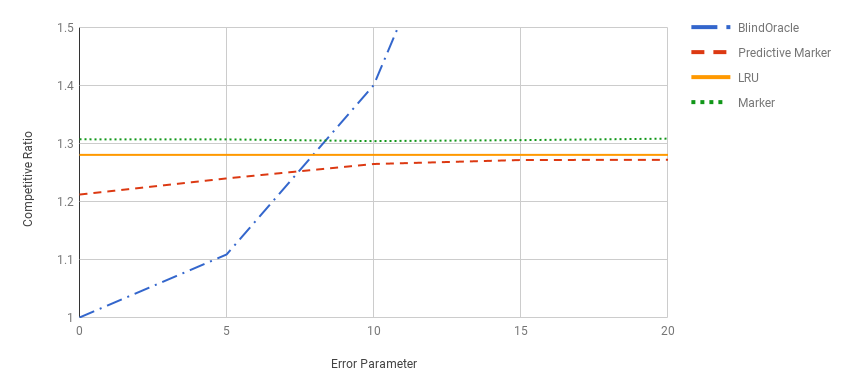}
\caption{Ratio of average number of evictions as compared to optimum for varying levels of oracle error.}
\label{fig:results}
\end{figure}

We set $k = 10$, and summarize the synthetic results on the \texttt{BK} dataset in Figure \ref{fig:results}. Observe that the performance of Predictive Marker is consistently better than LRU and standard Marker, and degrades slowly as the average error increases, as captured by the theoretical analysis. Second, we empirically verify that blindly following the oracle works well when the error is very low, but quickly becomes incredibly costly. 

\begin{table}[t!]
    \centering
    \begin{tabular}{c|c|c}
        Algorithm & Competitive Ratio on BK & Competitive Ratio on Citi \\ 
        \hline
        Blind Oracle & 2.049 & 2.023 \\
        LRU & 1.280 & 1.859 \\
        Marker & 1.310 & 1.869 \\
        Predictive Marker & 1.266 & 1.810 \\
    \end{tabular}
    \caption{Competitive Ratio using PLECO model.}
    \label{tab:results}
\end{table}

The results using the PLECO predictor are shown in Table \ref{tab:results}, where we keep $k = 10$ for the \texttt{BK} dataset and set $k = 100$ for \texttt{Citi}; we note that the ranking of the methods is not sensitive to the cache size, $k$. We can again see that the Predictive Marker algorithm outperforms all others, and is 2.5\% better than the next best method, LRU. While the gains appear modest, we note they are statistically significant at $p < 0.001$. Moreover, the off-the-shelf PLECO model was not tuned or optimized for predicting the {\em next} appearance of each element.

In that regard, the large difference in performance between using the predictor directly (Blind Oracle) and using it in combination with Marker (Predictive Marker) speaks to the power of the algorithmic method. By considering only the straightforward use of the predictor in the Blind Oracle setting, one may deem the ML approach not powerful enough for this application; what we show is that a more judicious use of the same model can result in tangible and statistically significant gains. 

\section{Conclusion}
\label{sec:conclusions}
In this work, we introduce the study of online algorithms with the aid of machine learned predictors. This combines the empirical success of machine learning with the rigorous guarantees of online algorithms. We model the setting for the classical caching problem and give an oracle-based algorithm whose competitive ratio is directly tied to the accuracy of the machine learned oracle. 

Our work opens up two avenues for future work. On the theoretical side, it would be interesting to see similar predictor
-based algorithms for other online settings such as the $k$-server problem; this has already led to a fruitful line of current research as we discussed in Section~\ref{ssec:related}. On the practical side, our caching algorithm shows how we can use machine learning in a safe way, avoiding problems caused by rare wildly inaccurate predictions. At the same time, our experimental results show that even with simple predictors, our algorithm provides an improvement compared to LRU. In essence, we have reduced the worst case performance of the caching problem to that of finding a good (on average) predictor. This opens up the door for practical algorithms that need not be tailored towards the worst-case or specific distributional assumptions, but still yield provably good performance.

\subsection*{Acknowledgements}
The authors would like to thank Andr\'es Mu\~noz-Medina and \'Eva Tardos for valuable discussions on the presentation of the paper, Shuchi Chawla and Seffi Naor for useful feedback regarding Section~\ref{ssec:rate_robustness}, Ola Svensson for suggesting the locality extension (Theorem~\ref{thm:locality})
as well as an anonymous reviewer for pointing towards the direction of Theorem \ref{thm:black_box}.

\bibliographystyle{alpha}
\bibliography{bib1}

\appendix

\section{Proof of Lemma \ref{lem:spread_common_losses}}\label{app:spread_common_losses}
In this section, we provide the proof of the lemma connecting spread to absolute and squared loss. Before doing so, we provide a useful auxiliary lemma.

\begin{lemma}\label{lem:AT_BT}
\vspace{0.1in}
For odd $T=2n+1$, one pair $(A_T,B_T)$ minimizing either absolute or squared loss subject to the constraints of the spread definition is $A_{2n+1}=(0\dots 2n)$ and $B_T=(n\dots n)$.
\end{lemma}
\begin{proof}
First we show that there exists a $B_T$ minimizing the loss with $b_i=b_j$ for all $i,j$. Assume otherwise; then there exist two subsequent $i,j$ with $b'_i> b'_j$. Since $a_i<a_j+1$ by the assumption on spread, $\min_{x\in{b_i,b_j}}\{\ell(a_i,b)+\ell(a_j,b)\}\leq \ell(a_i,b_i)+\ell(a_j,b_j)$. Applying this recursively, we conclude that such a $B_T$ exists.

Second, we show that there exist an $A_T$ that consists of elements $a_{i+1}=a_i+1$. Since the elements of $B_T$ are all equal to $b$, the sequence $\sum_{i=0}^{2n} \ell(a_i,b)$ is minimized for both absolute and squared loss when $a_{i}=b+i-n$.

Last, the exact value of $b$ does not make a difference and therefore we can set it to be $b=n$ concluding the lemma.
\end{proof}

\textbf{Lemma \ref{lem:spread_common_losses} restated:}
For absolute loss, $\ell_1(A,B) = \sum_i |a_i - b_i|$, the spread of $\ell_1$ is $S_{\ell_1}(m) \leq \sqrt{5m}$. \\
For squared loss, $\ell_2(A, B) = \sum (a_i - b_i)^2$, the spread of $\ell_2$ is $S_{\ell_2}(m) \leq \sqrt[3]{14m}.$ 

\begin{proof}
It will be easier to restrict ourselves to odd $T=2n+1$ and also assume that $T\geq 3$. This will give an upper bound on the spread (which is tight up to small constant factors). By Lemma \ref{lem:AT_BT},
 a pair of sequence minimizing absolute/squared loss is $A_T=(0,\dots,2n)$ and $B_T=(n,\dots,n)$. 
We now provide bounds on the spread based on this sequence, that is we find a $T=2n+1$ that
satisfies the inequality $\ell(A_T,B_T)\leq m$.

\emph{Absolute loss:} The absolute loss of the above sequence is:
$$
\ell(A_T,B_T)=2\cdot \sum_{j=1}^n j= 2\cdot \frac{n(n+1)}{2}=n(n+1)= \frac{T-1}{2}\cdot \frac{T+1}{2}=\frac{T^2-1}{4}.
$$
A $T$ that makes $\ell(A_T,B_T)\geq m$ is $T=\sqrt{4m+1}$. Therefore, for absolute loss $S_{\ell}(m)\leq \sqrt{5m}$, since $m\geq 1$

\emph{Squared loss:} The squared loss of the above sequence is:
$$
\ell(A_T,B_T)=2\cdot\sum_{j=1}^n j^2=2\cdot \frac{n(n+1)(2n+1)}{6}=\frac{(T^2-1)\cdot T}{12}=\frac{T^3-T}{12}\geq \frac{8T^3}{9\cdot 12}=\frac{2T^3}{27}
$$
where the inequality holds because $T\geq 3$.

A $T$ that makes $\ell(A_T,B_T)\geq m$ is $T=\sqrt[3]{14m}$. Therefore, for squared loss $S_{\ell}(m)\leq \sqrt[3]{14m}$.
 \end{proof}

\end{document}